\documentclass[letterpaper, 10pt, conference]{ieeeconf}
\IEEEoverridecommandlockouts
\pdfoutput=1

\usepackage{geometry}
\usepackage{graphicx}
\usepackage{epsfig}
\usepackage{amsmath}
\usepackage{amssymb} 
\usepackage[noadjust]{cite}
\usepackage{bm}
\usepackage{subfigure}
\usepackage{acronym}
\usepackage{paralist}
\usepackage{float}
\usepackage{epstopdf}
\usepackage{algorithm}
\usepackage{algpseudocode}
\usepackage{mathtools}
\usepackage{multicol} 
\usepackage{accents}
\usepackage{centernot}
\usepackage{xcolor}

\makeatletter
\def\BState{\State\hskip-\ALG@thistlm}
\makeatother

\newcommand{\limfy}[1]{\lim_{#1 \rarr \infty}}
\newcommand{\sfn}{S_f^\N}

\newcommand{\epso}{\overline{\epsilon}}
\newcommand{\epsu}{\underline{\epsilon}}

\newtheorem{problem}{Problem}

\newtheorem{define}{Definition}
\newtheorem{theorem}{Theorem}
\newtheorem{lemma}{Lemma}
\newtheorem{remark}{Remark}
\newtheorem{assume}{Assumption}

\newcommand{\Rc}{\mathcal R}
\newcommand{\R}{\mathbb R}

\newcommand{\N}{\mathcal{N}}
\newcommand{\D}{\mathcal{D}}
\newcommand{\V}{\mathcal{V}}

\newcommand{\G}{\mathcal{G}}
\newcommand{\E}{\mathcal{E}}

\newcommand{\Le}{\mathcal{L}}
\newcommand{\A}{\mathcal{A}}
\newcommand{\J}{\mathcal{J}}

\newcommand{\mb}{\overline{m}}
\newcommand{\Mb}{\overline{M}}
\newcommand{\Z}{\mathbb{Z}}

\newcommand{\pth}[1]{\left(#1\right)} 
\newcommand{\brc}[1]{\left \{#1\right \}} 

\DeclarePairedDelimiter{\ceil}{\lceil}{\rceil}
\DeclarePairedDelimiter{\floor}{\lfloor}{\rfloor}
\DeclarePairedDelimiter{\abs}{\lvert}{\rvert}

\newcommand{\rarr}{\rightarrow} 
 
\makeatletter
\let\oldceil\ceil
\def\ceil{\@ifstar{\oldceil}{\oldceil*}}
\makeatother
\makeatletter
\let\oldfloor\floor
\def\floor{\@ifstar{\oldfloor}{\oldfloor*}}
\makeatother
\makeatletter
\let\oldnorm\norm
\def\norm{\@ifstar{\oldnorm}{\oldnorm*}}
\makeatother
\makeatletter
\let\oldabs\abs
\def\abs{\@ifstar{\oldabs}{\oldabs*}}
\makeatother

\newtheorem{corollary}{Corollary}[theorem]

\begin{document}

\title{\LARGE \bf Resilient Leader-Follower Consensus to Arbitrary Reference Values in Time-Varying Graphs}

\author{James Usevitch and Dimitra Panagou
\thanks{James Usevitch is with the Department of Aerospace Engineering, University of Michigan, Ann Arbor; \texttt{usevitch@umich.edu}.
Dimitra Panagou is with the Department of Aerospace Engineering, University of Michigan, Ann Arbor; \texttt{dpanagou@umich.edu}.
The authors would like to acknowledge the support of the Automotive Research Center (ARC) in accordance with Cooperative Agreement W56HZV-14-2-0001 U.S. Army TARDEC in Warren, MI, and of the Award No W911NF-17-1-0526. This work has been funded by the Center for Unmanned Aircraft Systems (C-UAS), a National Science Foundation Industry/University Cooperative Research Center (I/UCRC) under NSF Award No. 1738714 along with significant contributions from C-UAS industry members. }
}

\maketitle
\thispagestyle{empty}
\pagestyle{empty}

\acrodef{wrt}[w.r.t.]{with respect to}
\acrodef{apf}[APF]{Artificial Potential Fields}
\begin{abstract}
	Several algorithms in prior literature have been proposed which guarantee consensus of normally behaving agents in a network that may contain adversarially behaving agents. These algorithms guarantee that the consensus value lies within the convex hull of initial normal agents' states, with the exact consensus value possibly being unknown. In leader-follower consensus problems however, the objective is for normally behaving agents to track a reference state that may take on values outside of this convex hull. In this paper we present methods for agents in time-varying graphs with discrete-time dynamics to resiliently track a reference state propagated by a set of leaders despite a bounded subset of the leaders and followers behaving adversarially. Our results are demonstrated through simulations.
\end{abstract}

\IEEEpeerreviewmaketitle

\section{Introduction}


Guaranteeing the resilience of multi-agent systems to adversarial misbehavior and misinformation is critically needed in modern autonomous systems. As part of this need, the \emph{resilient consensus problem} has been treated in the literature for several decades. In this problem, normally behaving agents in a multi-agent network seek to come to agreement on one or more state variables in the presence of adversarially behaving agents whose identity is unknown. Several algorithms based upon the \emph{Mean-Subsequence-Reduced} family of algorithms \cite{kieckhafer1994reaching} have been proposed which guarantee consensus of the normally behaving agents when the number of adversaries is bounded and the network communication structure satisfies certain \emph{robustness} properties. These discrete-time algorithms, which include the W-MSR, SW-MSR, DP-MSR,  and QW-MSR algorithms \cite{leblanc2013resilient,dibaji2017resilient,saldana2017resilient,dibaji2018resilient}, guarantee that the final consensus value of the normal agents is within the convex hull of the normal agents' initial states. {  However, the exact final value within this convex hull may depend in part upon the behavior of the adversarial agents.}

 A related problem in prior literature is the \emph{leader-follower consensus problem}, where the objective is for normally behaving agents to come to agreement on the reference value of a leader or set of leaders \cite{Dimarogonas2009leader,Ren2007,Ren2008consensus}. Prior work in this area typically assumes that there are no adversarially misbehaving agents; i.e. all leaders and followers follow the intended control laws.  An interesting direction of research is extending the property of resilience to the leader-follower consensus scenario, i.e., follower agents tracking the leader agents' reference state while rejecting the influence of adversarial agents whose identity is unknown. One aspect which prevents prior resilient consensus results from being extended to this case is that the reference state may not lie within the convex hull of normal agents' initial states.

In addition, leader-follower consensus can be viewed in an adversarial light. In some scenarios the objective of adversarial agents in a network may be to drive as many agents' states as possible towards a malicious value or unsafe set. The misbehaving nodes in a network may be considered ``leaders" seeking to divert normal agents' states to harmful values. Prior results in the literature on MSR-type algorithms guarantee that the adversaries cannot drive normal nodes' states to arbitrary values when the adversarial set is bounded, but do not analyze the extent of the adversaries' influence when these bounds are violated.

 Recent work related to the resilient leader-follower consensus problem includes \cite{leblanc2014,Mitra2016secure,mitra2018byzantine,usevitch2018resilient}. In \cite{leblanc2014}, the problem of resilient distributed estimation is considered where certain "reliable agents" drive the errors of the remaining normal agents to the static reference value of zero in the presence of misbehaving agents. In \cite{Mitra2016secure,mitra2018secure}, the problem of distributed, resilient estimation in the presence of misbehaving nodes is treated. The authors show conditions under which information about the decoupled modes of the system is resiliently transmitted from a group of source nodes to other nodes that cannot observe those modes. Their results guarantee exponential convergence to the reference modes of the system. In our prior work \cite{usevitch2018resilient}, we considered the case of leader-follower consensus to arbitrary static reference values using the W-MSR algorithm \cite{leblanc2013resilient}. In addition, the resilient leader-follower consensus problem is closely related to the \emph{secure broadcast} problem \cite{koo2006reliable2}, where an agent known as a ``dealer" seeks to broadcast a message to an entire network in the presence of misbehaving agents.

In this technical note, we briefly address the problem of resilient leader-follower consensus in the discrete-time domain. Specifically, we make the following contributions:
\begin{itemize}
	\item We demonstrate conditions under which normally behaving agents in time-varying graphs can resiliently track the reference signal of a set of leaders in the presence of a bounded number of arbitrarily misbehaving agents using the Sliding Window Mean-Subsequence-Reduced (SW-MSR) algorithm. To demonstrate these conditions, we introduce the notion of \emph{strong $(T,t_0,r)$-robustness with respect to a subset $S$}, which to the best of our knowledge has not been previously defined.
	\item We demonstrate sufficient conditions under which a properly selected subset of adversarially behaving agents can drive a network of agents appyling the SW-MSR algorithm to any arbitrary value.
\end{itemize}

This paper is organized as follows: in Section \ref{sec:notation} we outline notation used in this paper, in Section \ref{sec:probform} we give the problem formulation, in Section \ref{sec:results} we outline conditions for resilient leader-follower consensus in time-varying graphs, in Section \ref{sec:adv} we discuss the adversarial implications of the results on resilient leader-follower consensus, in Section \ref{sec:sim} we present simulations demonstrating our results, and in Section \ref{sec:concl} we give a conclusion and directions for future work.

\section{Notation}
\label{sec:notation}

The set of real numbers and integers are denoted $\R$ and $\Z$, respectively. The set of nonnegative reals and nonnegative integers are denoted $\R_+$ and $\Z_+$, respectively. 
The cardinality of a set $S$ is denoted $|S|$. The set union, intersection, and set difference operations of two sets $S_1$ and $S_2$ are denoted by $S_1 \cup S_2$, $S_1 \cap S_2$, and $S_1 \backslash S_2$ respectively. We denote $\bigcup_{i=1}^n S_i = S_1 \cup S_2 \cup \ldots \cup S_n$.
 A digraph of $n$ agents, with $n \in \Z_+$, is denoted as $\D[t] = (\V,\E[t])$ where $\V = \{1,\ldots, n\}$ is the set of labeled agents (represented by the vertices of the graph), and $\E[t] \subset \V \times \V$ is the (possibly time-varying) set of edges. An edge from $i$ to $j$, $i,j \in \V$, denoted as $(i,j) \in \E[t]$, represents the ability of the \emph{head} $i$ to send information to the \emph{tail} $j$ at time $t$. Note that for digraphs $(i,j) \in \E[t]$
does not necessarily imply that
$(j,i) \in \E[t]$. The set of in-neighbors of agent $i$ is denoted $\V_i[t] = \{j \in \V: (j,i) \in \E[t])\}$. Similar to \cite{leblanc2013resilient}, we define the inclusive neighbor set of node $i$ as $\J_i[t] = \V_i[t] \cup \{i\}$. The set of out-neighbors of each agent $i$ is denoted $\V_i^{out}[t] = \{k \in \V : (i,k) \in \E[t] \}$.

\section{Problem Formulation}
\label{sec:probform}

Consider a digraph of $n$ agents with time-varying edges, denoted $\D[t] = (\V, \E[t])$. Each agent $i \in \V$ has a state $x_i[t] \in \R$. Two types of agents are considered: leader agents (also called ``source" agents) and follower agents. The set of leader agents consists of agents which propagate a desired reference signal to the set of follower agents.
\begin{define}
\label{def:LsetSfset}
	The set of leader agents is denoted $\Le \subset \V$. The set of follower agents is denoted $S_f = \V \backslash \Le$.
\end{define}

\begin{assume}
\label{a:partition}
The sets $\Le$ and $S_f$ are static and satisfy $\Le \cup S_f = \V$ and $\Le \cap S_f = \emptyset$.
\end{assume}

Each normally behaving leader agent $l$ updates its state according to a reference function $f_r : \R \rarr \R$ as follows:
\begin{align}
\label{eq:leader}
	x_l[t+1] = f_r[t].
\end{align}
The precise definition of \emph{normally behaving} will be given in Definition \ref{def:normal}.

The purpose of this paper is to determine conditions under which normally behaving follower agents resiliently achieve consensus with a static reference state of the set of normally behaving leader agents in the presence of a possibly nonempty set of misbehaving agents, where the precise definition of misbehaving agents will be given in Definition \ref{def:misbehaving}.

\begin{problem}
	Given a digraph $\D[t] = (\V,\E[t])$ with a time-varying edge set satisfying Assumption \ref{a:partition}, determine conditions under which $\limfy{t} \max_{i,\ l} |x_i[t] - x_l[t]| = 0$ for all normally behaving follower agents $i$ and for all normally behaving leaders $l$ in the presence of a possibly nonempty misbehaving subset of agents $\A \subset \V$.
\end{problem}

Each normally behaving leader agent is able to send its state value to its out-neighbors at each time $t$.
In addition, each normally behaving follower agent $i \in S_f$ can receive state values from its in-neighbors at each time $t$, and can also send its own state value to its out-neighbors at each time $t$.
\begin{define}
	The value received by agent $i$ from agent $j$ at time $t$ is denoted $x_j^i[t]$.
\end{define}
Since the set of edges $\E[t]$ is time-varying, agents use a sliding-window approach over a time period $T \in \Z_+$ when taking into account information received from their in-neighbors. Let $T' = \min(T,t-t_0),\ t \geq t_0$. At each time $t \geq t_0$, each normally behaving follower agent $i$ considers information received from the set
\begin{align}
\label{eq:Ji}
	\J_i^T[t] &= \bigcup_{\tau = 0}^{T'} \J_i[t - \tau],		
\end{align}
i.e. the union of $i$'s in-neighbor sets over the time interval $[t-T,t]$ if $t \geq t_0 + T$, or $[t_0,t]$ if $t < t_0 + T$.
Each normally behaving follower agent $i$ then updates its state according to the \emph{Sliding Weighted Mean-Subsequence-Reduced} (SW-MSR) algorithm \cite{saldana2017resilient}, which is outlined in Algorithm \ref{alg:swmsr}. In essence, the SW-MSR algorithm causes normally behaving follower agents to update their state based on the most recently received information from each in-neighbor in $\J_i^T[t]$. In addition, agents filter out a subset of the information received based upon a prespecified parameter $F \in \Z_+$.

\begin{algorithm}
\caption{\small{\textsc{SW-MSR Algorithm} \cite{saldana2017resilient}:}}
\label{alg:swmsr}
\begin{enumerate}
	\item At each time step $t$, each agent $i$ forms a sorted list $\Omega_i[t]$ of the most recently received values from its in-neighbors as follows: { 
			\medmuskip=0mu
\thinmuskip=0mu
	\begin{align}	
		\tau_{ij}[t] &= \max\pth{\brc{\tau \in [t-T',t] : j \in \J_i[\tau]}},\ \forall j \in \J_i^{T'}[t] \nonumber \\
		\Omega_i[t] &= \{x_j^i[\tau_{ij}[t]] : j \in \J_i^{T'}[t]\}, \label{eq:tau_Omega}
	\end{align}
	}
	with $T' = \min(T,t-t_0)$ and $\J_i^T[t]$ defined in \eqref{eq:Ji}.\footnotemark
	\item If there are less than $F$ values strictly greater than $x_i[t]$ in $\Omega_i[t]$, then agent $i$ removes all values strictly greater than $x_i[t]$ from $\Omega_i[t]$. Otherwise, agent $i$ removes the $F$ largest values from $\Omega_i[t]$.
	\item \emph{In addition}, if there are less than $F$ values strictly less than $x_i[t]$ in $\Omega_i[t]$, then agent $i$ removes all values strictly less than $x_i[t]$ from $\Omega_i[t]$. Otherwise, agent $i$ removes the $F$ smallest values from $\Omega_i[t]$.
	\item Let $\Rc_i[t]$ denote the set of all agent indices whose state values were removed from $\Omega_i[t]$ in steps 2) and 3). Each normal agent $i$ applies the update
	\begin{align}
	\label{eq:follower}
		&x_i[t+1] = u_i[t] \\
		&u_i[t] = \sum_{j \in \J_i^{T'}[t] \backslash \Rc_i[t]} w_{ij}[t] x_j^i[\tau_{ij}[t]] 
	\end{align}
	
	where $\forall t$ and $\forall i \in S_f$ the weights satisfy $w_{ij}[t] \geq \alpha > 0$ $\forall j \in \J_i^{T'}[t]$, and $\sum_{j \in \J_i^{T'}[t] \backslash \Rc_i[t]}w_{ij}[t]$ = 1.
\end{enumerate}
\end{algorithm}

\begin{remark}
	If $T = 0$, the SW-MSR algorithm essentially reduces to the \emph{Weighted Mean-Subsequence-Reduced} (W-MSR) algorithm \cite{leblanc2013resilient}. The SW-MSR algorithm can be seen as a generalization of the W-MSR algorithm to digraphs with time-varying edge sets.
\end{remark}

In contrast to much of the prior literature on leader-follower consensus which typically assumes that all agents apply nominally specified control laws, this paper considers the presence of \emph{misbehaving agents}:
\begin{define}
\label{def:misbehaving}
    An agent $j \in \V$ is \emph{misbehaving} if at least one of the following conditions hold:
    \begin{enumerate}
        \item There exists a time $t$ where agent $j$ does not update its state according to \eqref{eq:leader} and also does not update its state according to \eqref{eq:follower}.
        \item There exists a time $t$ where $j$ does not communicate its true state value $x_j(t)$ to at least one of its out-neighbors; i.e. $\exists t \geq t_0$ and $\exists k \in V_j^{out}[t]$ s.t. $x_j[t] \neq x_j^k[t]$.
        \item There exists a time $t$ where $j$ communicates different values to different out-neighbors; i.e. $\exists t \geq t_0$ and $\exists k_1,k_2 \in V_j^{out}[t]$ s.t. $x_j^{k_1}[t] \neq x_j^{k_2}[t]$.
\end{enumerate}
The set of misbehaving agents is denoted $\A \subset \V$.
\end{define}

\begin{define}
\label{def:normal}
	The set of agents which are \emph{not} misbehaving are denoted $\N = \V \backslash \A$. Agents in $\N$ are referred to as \emph{normally behaving agents}.
\end{define}
Intuitively, misbehaving agents are agents which update their states arbitrarily or communicate false information to their out-neighbors. By Definition \ref{def:misbehaving}, the set of misbehaving agents $\A$ includes both malicious agents and Byzantine agents \cite{leblanc2013resilient}.

This paper considers scenarios where both followers \emph{and leaders} are vulnerable to adversarial attacks and faults, and therefore the set $\A \cap \Le$ may possibly be nonempty, and the set $\A \cap S_f$ may possibly be nonempty. The following notation will be used:
\begin{define}[Misbehaving agent notation]
The set of misbehaving leaders is denoted as $\Le^\A = \Le \cap \A$. The set of misbehaving followers is denoted as $S_f^\A = S_f \cap \A$.
\end{define}

\begin{define}[Normally behaving agent notation]
\label{def:normalsets}
The set of normally behaving leaders is denoted $\Le^\N = \Le \backslash \A$. The set of normally behaving followers is denoted $S_f^\N = S_f \backslash \A$. 
\end{define}

\footnotetext{Observe that by the definition of $\J_i[t]$, $x_i^i[\tau_{ii}[t]] \in \Omega_i[t]$ for all $t \geq t_0$. This implies that the set $\Omega_i[t]$ is never empty at any time, even for $t_0 \leq t < t + T$.}

\subsection{Review of Resilient Consensus Concepts}

This subsection will briefly review several definitions associated with the resilient consensus literature that will be used in this paper.
To quantify the distribution of agents in $\A$ throughout a digraph $\D$, the notions of $F$-total and $F$-local sets are used.

\begin{define}[\cite{leblanc2013resilient}]
Let $F \in \Z_+$. A set $S \subset \V$ is \emph{F-total} if it contains at most $F$ nodes; i.e. $|S| \leq F$.
\end{define}

\begin{define}[\cite{leblanc2013resilient}]
Let $F \in \Z_+$. A set $S \subset \V$ is \emph{F-local} with respect to (w.r.t.) a given $t_0 \in \Z$ if $|S \cap \V_i[t]| \leq F$ $\forall i \in \V \backslash S$, $\forall t \geq t_0$. 
\end{define}

Sufficient conditions for the success of several resilient consensus algorithms involve the graph theoretical notions of $r$-reachability, $r$-robustness, and strong $r$-robustness, which are defined as follows:

\begin{define}[\cite{leblanc2013resilient}]
\label{def:rreach}
Let $r \in \Z_+$ and $\D=(\V,\E)$ be a digraph. A nonempty subset $S \subset \V$ is $r$-reachable if $\exists i \in S$ such that $|\V_i \backslash S| \geq r$.
\end{define}
\begin{define}[\cite{leblanc2013resilient}]
\label{def:rrobust}
Let $r \in \Z_+$. A nonempty, nontrivial digraph $\D = (\V,\E)$ on $n$ nodes $(n \geq 2)$ is $r$-robust if for every pair of nonempty, disjoint subsets of $\V$, at least one of the subsets is $r$-reachable. By convention, the empty  graph $(n = 0)$ is 0-robust and the trivial graph $(n=1)$ is 1-robust.
\end{define}

\begin{define}[Strong $r$-robustness w.r.t. $S$ \cite{Mitra2016secure}]
\label{def:strongr}
Let $r \in \Z_+$, $\D=(\V,\E)$ be a digraph, and $S \subset \V$ be a nonempty subset. $\D$ is strongly $r$-robust w.r.t. $S$ if for any nonempty subset $C \subseteq \V \backslash S$, $C$ is $r$-reachable.
\end{define}

\begin{remark}
	Given a particular subset $S \subset \V$, it can be verified in polynomial time whether $\D$ is strongly robust w.r.t. $S$ \cite{mitra2018byzantine}. {  On the other hand, determining the $r$-robustness of a digraph is NP-hard in general \cite{leblanc2013algorithms}, but can be computed using mixed integer linear programming \cite{usevitch2019determining, usevitch2019MILP}.}
\end{remark}

In this paper, we introduce the concept of \emph{strong $(T,t_0,r)$-robustness}, which is defined as follows:

\begin{define}
    Let $T,r \in \Z_{\geq 0}$ and let $t_0 \in \Z$. Let $\D[t] = (\V,\E[t])$ be a digraph with a time-varying edge set, and define $\D^T[t] = \bigcup_{\tau = 0}^T \D[t-\tau]$. Then $\D[t]$ is strongly $(T,t_0,r)$-robust with respect to a subset $S \subset \V$ if $\D^T[t]$ is strongly $r$-robust with respect to $S \subset \V$ for all $t \geq t_0+T$.
\end{define}

\begin{remark}
Strong $(T,t_0,r)$-robustness generalizes the notion of strong $r$-robustness to digraphs with a time-varying edge set. Note that the property of strong $r$-robustness in Definition \ref{def:strongr} is a particular case of strong $(T,t_0,r)$-robustness with $T = 0$.
\end{remark}

\begin{remark}
	In many practical networks it may be difficult to ensure that a digraph $\D[t]$ is strongly $r$-robust w.r.t. $S$ at every time step $t$. The time window $T$ relaxes this requirement by only requiring the union of $\D[t]$ over the last $T$ timesteps to be strongly $r$-robust w.r.t. $S$. Increasing $T$ allows for edges to be ``active" less often while still preserving the $(T,t_0,r)$-robustness of $\D[t]$.
\end{remark}

\section{Resilient Leader-Follower Consensus in Time-Varying Graphs}
\label{sec:results}

For our analysis of time-varying graphs, the following functions are defined (with $T' = \min(T, t-t_0)$ as per Algorithm \ref{alg:swmsr}):
\begin{align}
	\Mb[t] &= \max_{i \in \sfn, l \in \Le^\N,  \tau \in [0,T']} (x_i[t-\tau],x_l[t-\tau]) \nonumber \\
    \mb[t] &= \min_{i \in \sfn, l \in \Le^\N, \tau \in [0,T']} (x_i[t-\tau],x_l[t-\tau]) \nonumber \\
    V[t] &= \Mb[t] - \mb[t] \label{eq:minmaxv}
\end{align}

The following Lemma establishes that $\Mb[t]$ and $\mb[t]$ are nonincreasing and nondecreasing functions, respectively, on any time interval where $f_r[t]$ is constant.

\begin{lemma}
    \label{lem:swmsrint}
Let $\D[t] = (\V,\E[t])$ be a nonempty, nontrivial, simple digraph with $S_f$ nonempty. Let $F,\tau \in \Z_+$, $t_0,t_1,t_2 \in \Z$ with  $t_2 > t_1 \geq t_0$.
Suppose that $\A$ is an $F$-local set with respect to $t_0$,
and suppose that all normally behaving agents $i \in \sfn$ apply the SW-MSR algorithm with parameter $F$. If $f_r[t]$ is constant $\forall t \in [t_1,t_2)$, then all of the following statements hold $\forall t \in [t_1,t_2)$:
\begin{itemize}
        \item $x_i[t] \in [\mb[t_1],\Mb[t_1]]\ \forall i \in \sfn$
        \item $\Mb[t]$ and $\mb[t]$ are nonincreasing and nondecreasing, respectively.
    \end{itemize}
\end{lemma}

\begin{proof}
    First, observe that $x_l[t] = f_r[t]$ $\forall l \in \Le^\N$, $\forall t \geq t_0$ by \eqref{eq:leader}. Since $f_r[t]$ is constant $\forall t \in [t_1,t_2)$, by \eqref{eq:minmaxv} we have $x_l[t] \in [\mb[t_1],\Mb[t_1]]$ $\forall l \in \Le^\N$, $\forall t \in [t_1,t_2)$. Next, consider any $i \in \sfn$. By definition of $\Mb[t]$ and $\mb[t]$, $\forall j \in \J_i^T[t_1] \backslash \A,\ x_j^i[\tau_{ij}[t_1]] \in [\mb[t_1],\Mb[t_1]]$ where $\tau_{ij}[t]$ is defined by \eqref{eq:tau_Omega}. Now consider any agent $k \in \A$. If we have $x_k^i[\tau_{ik}[t_1]] > \Mb[t_1] \geq x_j(\tau_{ij}[t_1])\ \forall j \in \V \backslash \A$, the fact that $|\A| \leq F$ implies any value $x_k^i[\tau_{ik}[t_1]]$ satisfying this condition is one of the $F$ highest values in $\Omega_i[t_1]$ and will be filtered out as per the SW-MSR Algorithm (Algorithm \ref{alg:swmsr}). Similarly, if $x_k^i[\tau_{ik}[t_1]] < \mb[t_1] \leq x_j(\tau_{ij}[t_1])\ \forall j \in \V \backslash \A$, then $x_k^i[\tau_{ik}[t_1]]$ is one of the $F$ lowest values in $\Omega_i[t_1]$ and will be filtered out. Therefore all state values in $\J_i^T[t_1] \backslash \Rc_i[t_1]$ fall in the interval $[\mb[t_1],\Mb[t_1]]\ \forall i \in \sfn$. Since the values of $w_{ij}[t_1]$ imply a convex combination of values in the set $\J_i[^T[t_1] \backslash \Rc_i[t_1]$, $x_i[t_1+1] \in [\mb[t_1],\Mb[t_1]]$. Further, since by definition of $\mb$ and $\Mb$ we have $x_i[t-\tau] \in [\mb[t_1],\Mb[t_1]]$ $\forall i \in \sfn,\ \forall l \in \Le^\N$, $\forall \tau \in [0,T']$ where $T' = \min(T,t-t_0)$, it holds that $x_i[t_1+1-\tau] \in [\mb[t_1],\Mb[t_1]]\ \forall \tau \in [0,T'],\ \forall i \in \sfn$, $\forall l \in \Le$. 
These arguments imply $\Mb[t_1+1] \leq \Mb[t_1]$. Similar arguments can be used to show $\mb[t_1+1] \geq \mb[t_1]$.

Now by induction assume $\Mb[t_1+p] \leq \Mb[t_1 + p-1]$ and $\mb[t_1+p] \geq \mb[t_1 + p-1]$, for all $p \in \Z_+$ such that $p \geq 1$, $t_1 + p < t_2-1$. By \eqref{eq:minmaxv}, $x_i[t] \in [\mb[t_1 + p],\Mb[t_1 + p]]$ $\forall i \in \sfn$, $\forall t \in [t_1 + p -T,t_1+p]$. In addition, $f_r[t]$ being constant on $[t_1,t_2)$ implies $x_l[t] \in [\mb[t_1 + p],\Mb[t_1 + p]]$ $\forall l \in \Le^\N$. Therefore $x_j^i[\tau_{ij}[t_1+p]] \in [\mb[t_1],\Mb[t_1]]$ $\forall j \in \J_i^T[t_1+p] \backslash \A$. Since $|\A| \leq F$, it can be shown by prior arguments that $x_j^i[\tau_{ij}[t_1+p]]$ for all $j \in \J_i^T[t_1+p] \backslash \Rc_i[t_1+p]$ will lie in the interval $[\mb[t_1 + p],\Mb[t_1 + p]]$ $\forall i \in \sfn$. Therefore all $i \in \sfn$ will update their states with a convex combination of values in $[\mb[t_1 + p],\Mb[t_1 +p]]$, implying $\mb[t_1 + p + 1] \geq \mb[t_1 + p]$ and $\Mb[t_1 + p + 1] \leq \Mb[t_1 + p]$.
\end{proof}

The next theorem demonstrates that the error between the normal nodes and normally behaving leaders decreases exponentially on any time interval $t \in [t_1,t_2)$ where $f_r[t]$ is constant and $t_2-t_1$ is sufficiently large.

\begin{theorem}
    \label{thm:swmsr}
    Let $\D[t] = (\V,\E[t])$ be a nonempty, nontrivial, simple digraph.
    Let $\Le, S_f, S_f^\N, \A$ be defined as per Definitions \ref{def:LsetSfset} and \ref{def:misbehaving}.
    Let $F \in \Z_+$, $t_0,t_1,t_2 \in \Z$ with $t_2 > t_1 \geq t_0 + T$, and let $V[t]$ be defined as in \eqref{eq:minmaxv}.
Suppose that $S_f$ is nonempty, $\A$ is an $F$-local set with respect to $t_0$, $\D[t]$ is strongly $(T,t_0,2F+1)$-robust w.r.t. the set $\Le$, and all normally behaving agents $i \in \sfn$ apply the SW-MSR algorithm with parameter $F$. If $f_r[t]$ is constant $\forall t \in [t_1-T,t_2)$ and $t_2 > t_1 + ( |\sfn| +1) \sigma T$ for some $\sigma \in \Z_+$, then
    \begin{align}
    	V[t_1+( |\sfn| +1) \sigma T] &\leq (1-\alpha^{(|\sfn|+1)T})^\sigma V[t_1 + T], \nonumber
    \end{align}
    where $0 < \alpha < 1$ is defined in Algorithm \ref{alg:swmsr}.  Furthermore, if $t_2 = \infty$ then
    \begin{align*}
    \limfy{t} V[t] = \limfy{t} \max_{i \in \sfn,\ l \in \Le^\N} |x_i[t] - x_l[t]| = 0.
    \end{align*}
    
\end{theorem}

\begin{proof}
Consider the case where $f_r[t]$ is constant for $t \in [t_1-T,t_2)$ and $t_2 < \infty$. This implies $x_l[t] = f_r[t]$ is constant $\forall l \in \Le^\N$, $\forall t \in [t_1-T,t_2)$.
We define
    \begin{align*}
        X_m(t,t',\epsu) &= \{i \in \N : x_i[t'-\tau] < \mb[t] + \epsu \nonumber\\ &\text{ for some } 0 \leq \tau \leq T\},  \nonumber\\
        X_M(t,t',\epso) &= \{i \in \N : x_i[t'-\tau] > \Mb[t] - \epso \nonumber\\ &\text{ for some } 0 \leq \tau \leq T\} , \nonumber\\
        S_X(t,t',\epsu,\epso) &= X_m(t,t',\epsu) \cup X_M(t,t',\epso), \nonumber \\
        \overline{S}_X(t,t',\epsu,\epso) &= \V \backslash S_X(t,t',\epsu,\epso).
    \end{align*}
We prove the result by first showing that $|S_X(t,t',\epsu,\epso)|$ decreases over an appropriate sequence of $t'$ and with an appropriate choice of $\epsu,\epso$.
Let $\epsu = f_r[t_1] - \mb[t_1]$ and $\epso = \Mb[t_1] - f_r[t_1]$.
$\D[t]$ is strongly $(T,t_0,2F+1)$-robust w.r.t $\Le$, implying $\D^T[t]$ is strongly $(2F+1)$-robust w.r.t $\Le$ $\forall t \geq t_0 + T$. $\D^T[t]$ being strongly $(2F+1)$-robust with respect to $\Le$ implies there exists a nonempty $S_1 \subseteq \sfn \subset \V \backslash \Le$ such that $\forall i_1 \in S_1$, $|\J_{i_1}^T[t_1]\backslash \sfn| \geq 2F+1$. Since $\A$ is $F$-local and $\V \backslash \sfn = \Le \cup \A$, this implies $|\J_{i_1}^T[t_1] \cap \Le^\N| \geq F+1$. This implies by the SW-MSR Algorithm, $\J_{i_1}^T[t_1] \backslash \Rc_{i_1}[t_1]$ contains at least one normally behaving leader $l \in \Le^\N$ with $x_l^i[\tau_{i_1 l}[t_1]] = f_r[t_1]$. This can be seen by noting that $x_l[t] = f_r[t]$ $\forall l \in \Le^\N$, $\forall t \in [t_1-T,t_2)$.
 Using this fact, lower bounds on $x_{i_1}[t]$ for all $i_1 \in S_1$ and $t \in [t_1+1,t_1+T]$ can be established as follows: recall that the weights $w_{ij}$ are lower bounded by $\alpha > 0$. By Lemma \ref{lem:swmsrint}, $x_j^{i_1}[t] \in [\mb[t_1],\Mb[t_1]]$ $\forall j \in \J_{i_1}^T[t] \backslash \Rc_{i_1}[t]$, $\forall t \in [t_1-T,t_2)$. Observe that
    \begin{align}
    	x_{i_1}[t_1+1] &= \sum_{j \in \J_{i_1}^T[t] \backslash \Rc_{i_1}[t]} w_{{i_1}j}[t] x_j^{i_1}[\tau_{{i_1}j}[t]], \\ 
        &\geq \alpha f_r[t_1] + (1-\alpha)\mb[t_1]. \label{eq:absmin}
    \end{align}
    Since there exists at least one normally behaving leader in $\J_{i_1}^T[t_1] \backslash \Rc_{i_1}[t_1]$, \eqref{eq:absmin} represents the minimum possible value for $x_{i_1}[t_1 +1]$. Extending these bounds to time $t_1 + T$ yields
    \begin{align}
    	x_{i_1}[t_1 + 2] &\geq \alpha x_{i_1}[t_1 + 1] +  (1-\alpha)\mb[t_1], \nonumber \\
    	        &\geq \alpha^2 f_r[t_1] + (1+\alpha)(1-\alpha)\mb[t_1], \nonumber \\
    	                x_{i_1}[t_1 + 3] 
        &\geq \alpha^3 f_r[t_1] + (1+\alpha + \alpha^2)(1-\alpha)\mb[t_1], \nonumber \\
        \vdots \qquad & \qquad \qquad \qquad \vdots \nonumber \\
        x_{i_1}[t_1 + k] &\geq \alpha^k f_r[t_1] + \bigg(\sum_{j=0}^{k-1} \alpha^j\bigg)(1-\alpha)\mb[t_1], \nonumber \\
        &\geq \alpha^k f_r[t_1] + (1-\alpha^k)\mb[t_1], \nonumber \\
        &\geq \mb[t_1] + \alpha^k\epsu. \label{eq:low1}
    \end{align}
    This holds for $0 < k \leq T$. 
	Using similar arguments, an upper bound on $x_{i_1}[t_1 + k]$ can be established as follows:
    \begin{align}
        x_{i_1}[t_1 +1] &\leq \alpha f_r[t_1] + (1-\alpha)\Mb[t_1], \nonumber\\
        \vdots \qquad & \qquad \qquad \qquad \vdots \nonumber \\
        x_{i_1}[t_1 +k] &\leq \alpha^k f_r[t_1] + (1-\alpha^k)\Mb[t_1],  \nonumber \\
        &\leq \Mb[t_1] - \alpha^k\epso, \label{eq:high1}
    \end{align}
    for $0 < k \leq T$. 
    Therefore $x_{i_1}[t_1 + T] \in [\mb[t_1] + \alpha^T\epsu, \Mb[t_1] - \alpha^T \epso]$ for all $i_1 \in S_1$.

We show next that $|S_X(t_1,t_1+2T,\alpha^{2T}\epsu, \alpha^{2T} \epso)| < |\sfn|$. Define $C_2$ as the set of all $i_2 \in \sfn$ such that $x_{i_2}[t_1 + T] \in [\mb[t_1] + \alpha^T\epsu, \Mb[t_1] - \alpha^T \epso]$. Since $S_1 \subseteq C_2$ by \eqref{eq:low1} and \eqref{eq:high1}, $C_2$ is therefore nonempty. Since each agent in $\sfn$ always uses its own state in \eqref{eq:follower} as per the SW-MSR algorithm, lower bounds on the state of each $i_2 \in C_2$ can be established as:
\begin{align}
	x_{i_2}[t_1+T+1] &\geq \alpha x_{i_2}[t_1+T] + (1-\alpha)\mb[t_1], \nonumber \\
	&\geq \alpha(\mb[t_1] + \alpha^T\epsu) + (1-\alpha)\mb[t_1], \nonumber \\
	&\geq \alpha^{T+1}\epsu + \mb[t_1] \nonumber \\
	x_{i_2}[t_1+T+2] &\geq \mb[t_1]+ \alpha^{T+2}\epsu , \nonumber \\
	\vdots \qquad & \qquad \qquad \qquad \vdots \nonumber\\
	x_{i_2}[t_1+T+k] &\geq  \mb[t_1] + \alpha^{T+k}\epsu,
\end{align}
which holds for $0 < k \leq T$. Similarly, the following upper bounds can be established:
\begin{align}
x_{i_2}[t_1+T+1] &\leq \alpha x_{i_2}[t_1+T] + (1-\alpha)\Mb[t_1], \nonumber \\
&\leq \alpha(\Mb[t_1] - \alpha^T\epso) + (1-\alpha)\Mb[t_1], \nonumber \\
&\leq \Mb[t_1] + \alpha^{T+1}\epso, \nonumber \\
x_{i_2}[t_1+T+2] &\leq  \Mb[t_1] + \alpha^{T+2}\epso, \nonumber \\
\vdots \qquad & \qquad \qquad \qquad \vdots \nonumber\\
x_{i_2}[t_1+T+k] &\leq \Mb[t_1] + \alpha^{T+k}\epso, \label{eq:lowbound2}
\end{align}
which holds for $0 < k \leq T$. These arguments imply that for all $i_2 \in C_2$, $i_2 \notin S_X(t_1,t_1+2T,\alpha^{2T}\epsu,\alpha^{2T}\epso)$. Therefore $|S_X(t_1,t_1+2T,\alpha^{2T}\epsu,\alpha^{2T}\epso)| < |\sfn|$.

We next show that $|S_X(t_1,t_1 + 3T,\alpha^{3T}\epsu,\alpha^{3T}\epso)| < |S_X(t_1,t_1 + 2T,\alpha^{2T}\epsu,\alpha^{2T}\epso)|$.
Since $\D^T[t]$ is strongly $(T,t_0,2F+1)$-robust, there exists a nonempty $S_3 \subseteq S_X(t_1,t_1+2T,\alpha^{2T}\epsu,\alpha^{2T}\epso)$ such that for all $i_3 \in S_3$, $|\J_{i_3}^T[t_1+2T] \cap \overline{S}_X(t_1,t_1+2T,\alpha^{2T}\epsu,\alpha^{2T}\epso)| \geq 2F+1$. Since $\A$ is an $F$-local set, $\J_{i_3}^T[t_1+2T] \cap \overline{S}_X(t_1,t_1+2T,\alpha^{2T}\epsu,\alpha^{2T}\epso)$ includes at least $F+1$ normally behaving agents from $\N$ $\forall i_3 \in S_3$. Observe that by the definition of $S_X(t_1,t_1+2T,\alpha^{2T}\epsu,\alpha^{2T}\epso)$ the state of each $i_3 \in S_3$ satisfies either $x_{i_3}[t_1+2T] < x_{j}^{i_3}[\tau_{i_3 j}[t_1+2T]]$ or $x_{i_3}[t_1+2T] > x_{j}^{i_3}[\tau_{i_3 j}[t_1+2T]]$ for all $j \in \N \cap \overline{S}_X(t_1,t_1+2T,\alpha^{2T}\epsu,\alpha^{2T}\epso)$. Therefore $i_3$ will incorporate at least one in-neighbor's state from the interval $[\mb[t_1] + \alpha^{2T}\epsu, \Mb[t_1] - \alpha^{2T}\epso]$ in its state update, yielding the following bounds for all $i_3 \in S_3$:
\begin{align}
	x_{i_3}[t_1+2T+1] &\geq \alpha(\mb[t_1] + \alpha^{2T}\epsu) + (1-\alpha)\mb[t_1] \nonumber \\
	&\geq \mb[t_1] + \alpha^{2T+1}\epsu \nonumber \\
	x_{i_3}[t_1+2T+2] &\geq \mb[t_1] + \alpha^{2T+2}\epsu	\nonumber \\
	\vdots \qquad & \qquad \qquad \qquad \vdots \nonumber\\
	x_{i_3}[t_1+2T+k] &\geq \mb[t_1] + \alpha^{2T+k}\epsu, \label{eq:lowbound2}
\end{align}
for all $0 < k \leq T$. Similarly,
\begin{align}
	x_{i_3}[t_1+2T+1] &\leq \alpha(\Mb[t_1] + \alpha^{2T}\epso) + (1-\alpha)\Mb[t_1] \nonumber \\
	&\leq \Mb[t_1] + \alpha^{2T+1}\epsu \nonumber \\
	x_{i_3}[t_1+2T+2] &\leq \Mb[t_1] + \alpha^{2T+2}\epso \nonumber \\
	\vdots \qquad & \qquad \qquad \qquad \vdots \nonumber\\
	x_{i_3}[t_1+2T+k] &\leq \Mb[t_1] + \alpha^{2T+k}\epso \label{eq:upbound2}
\end{align}
for all $0 < k \leq T$. 
This implies that $i_3 \notin S_X(t_1,t_1 + 3T,\alpha^{3T}\epsu,\alpha^{3T}\epso)$ $\forall i_3 \in S_3$. 
Furthermore, we define $C_3$ as the set of all $j_3 \in \sfn$ such that $x_{j_3}[t_1 + 2T] \in [\mb[t_1] + \alpha^{2T}\epsu, \Mb[t_1] -\alpha^{2T}\epso ]$. By this definition, $C_2 \subseteq C_3$.
Note that the bounds in equations \eqref{eq:lowbound2} and \eqref{eq:upbound2} also apply to all agents $j_3 \in C_3$ since $x_{j_3}[t_1+2T] \in [\mb[t_1] + \alpha^{2T}\epsu, \Mb[t_1] - \alpha^{2T}\epso]$ $\forall j_3 \in C_3$, and each $j_3$ does not filter out its own state. Therefore $j_3 \notin S_X(t_1,t_1 + 3T,\alpha^{3T}\epsu,\alpha^{3T}\epso)$ $\forall j_3 \in C_3$, and therefore $|S_X(t_1,t_1 + 3T,\alpha^{3T}\epsu, \alpha^{3T}\epso)| < |S_X(t_1,t_1 + 2T,\alpha^{2T}\epsu,\alpha^{2T}\epso)|$.

This logic can be continued iteratively to show that $|S_X(t_1,t_1 + pT,\alpha^{pT}\epsu, \alpha^{pT}\epso)| < |S_X(t_1,t_1 + (p-1)T,\alpha^{(p-1)T}\epsu,\alpha^{(p-1)T}\epso)|$ for all $p \geq 2$, $p \in \Z$ such that $t_1 + pT < t_2$.
This can be done by defining
{\medmuskip=0mu
\thinmuskip=0mu
\thickmuskip=0mu
\begin{align}C_{p} \hspace{.5em} = 
\{i_{p} \in \sfn : x_{i_{p}}[t_1 + (p-1)T] \in [\mb[t_1] + \alpha^{(p-1)T} \epsu,& \nonumber \\
\Mb[t_1] + \alpha^{(p-1)T}\epso]\}&, \nonumber
\end{align}
}
which satisfies $C_{p-1} \subseteq C_{p}$,
and considering each $S_X(t_1,t_1 + (p-1)T,\alpha^{(p-1)T}\epsu,\alpha^{(p-1)T}\epso)$ for $p \geq 3$. Since $\D^T[t]$ is $(T,t_0,2F+1)$-robust, if $S_X(t_1,t_1 + (p-1)T,\alpha^{(p-1)T}\epsu,\alpha^{(p-1)T}\epso)$ is nonempty at time $t_1 + (p-1)T$ then there exists a nonempty $S_p \subseteq S_X(t_1,t_1 + (p-1)T,\alpha^{(p-1)T}\epsu,\alpha^{(p-1)T}\epso)$ such that $\forall i_p \in S_p$, $|\J_{i_p}^T[t_1 + (p-1)] \cap \overline{S}_X(t_1,t_1 + (p-1)T,\alpha^{(p-1)T}\epsu,\alpha^{(p-1)T}\epso)| \geq 2F+1$. Using prior arguments, it can then be shown that $x_{i_p}[t_1 + pT] \in [\mb[t_1] + \alpha^p \epsu, \Mb[t_1] - \alpha^p \epso]$. This implies that $i_p \notin S_X(t_1,t_1 + pT,\alpha^{pT}\epsu,\alpha^{pT}\epso)$ $\forall i_p \in S_p$. Similarly, by using prior arguments it also holds that $x_{j_{p}}[t_1 + pT] \in [\mb[t_1] + \alpha^p \epsu, \Mb[t_1] - \alpha^p \epso]$ $\forall j_p \in C_p$, and therefore $j_{p} \notin S_X(t_1,t_1 + pT,\alpha^{pT}\epsu, \alpha^{pT}\epso )$ $\forall j_{p} \in C_p$. This implies that $|S_X(t_1,t_1 + pT,\alpha^{pT}\epsu, \alpha^{pT}\epso )| < |S_X(t_1,t_1 + (p-1)T,\alpha^{(p-1)T}\epsu, \alpha^{(p-1)T}\epso)|$ for all $p \geq 2$, $p \in \Z$ such that $t_1 + pT < t_2$.

Since $\sfn \subset \V$ is finite, there exists a $p' > 1$, $p' \in \Z_+$ such that $S_X(t_1,t_1 + (p'+1)T,\alpha^{(p'+1)T}\epsu,\alpha^{(p'+1)T}\epso) = \emptyset$. 
 This implies that for all $i \in \sfn$,
\begin{align}
	 x_i[t_1 + (p'+1)T] &\geq \mb[t_1] + \alpha^{(p'+1)T} \epsu \nonumber \\
	 x_i[t_1 + (p'+1)T] &\leq \Mb[t_1] + \alpha^{(p'+1)T} \epso
\end{align}
Considering $V[t_1 + (p'+1)T]$, we have
\begin{align}
	&V[t_1 + (p'+1)T] = \nonumber \\
	&\hspace{3em}\Mb[t_1 + (p'+1)T] - \mb[t_1 + (p'+1)T] \nonumber \\
	&\leq \Mb[t_1] - \alpha^{(p'+1)T} \epso - (\mb[t_1] + \alpha^{(p'+1)T} \epsu) \nonumber \\
	&\leq V[t_1] - \alpha^{(p'+1)T}(\epsu + \epso)
\end{align}
Recall that $\epsu = f_r[t_1] - \mb[t_1]$ and $\epso = \Mb[t_1] - f_r[t_1]$. This implies that $\epsu + \epso = \Mb[t_1] - \mb[t_1] = V[t_1]$, implying
\begin{align}
\label{eq:someeq}
	&V[t_1 + (p'+1)T] \leq \nonumber \\
	&\hspace{3em} V[t_1] - \alpha^{(p'+1)T}V[t_1] = (1 - \alpha^{(p'+1)T})V[t_1]
\end{align}
Recalling that $|S_X(t_1,t_1+2T,\alpha^{2T}\epsu, \alpha^{2T} \epso)|$ $< |\sfn|$ at time $t_1 + 2T$, and that $|S_X(t_1,t_1 + pT,\alpha^{pT}\epsu, \alpha^{pT}\epso)|$ $< |S_X(t_1,t_1 + (p-1)T,\alpha^{(p-1)T}\epsu, \alpha^{(p-1)T}\epso)|$ for all $p \geq 3$,
it follows that 
$p' \leq |\sfn|$
since $S_X(t_1,t_1 + (p'+1)T,\alpha^{(p'+1)T}\epsu,\alpha^{(p'+1)T}\epso) = \emptyset$ after no more than $(|\sfn| + 1)T$ time steps.
Therefore we have $V[t_1 + (|\sfn|+1)T] \leq (1 - \alpha^{(|\sfn|+1)T})V[t_1]$ by substituting $p' = |\sfn|$ into \eqref{eq:someeq}.
The above analysis can be repeated to show
{\medmuskip=0mu
\begin{align}
V[t_1 + (|\sfn|+1)\sigma T] \leq (1 - \alpha^{(|\sfn|+1)T})V[t_1 + (\sigma-1)T] \nonumber
\end{align}
}
for $\sigma \geq 1$, $\sigma \in \Z$ such that $t_1 + (|\sfn|+1)\sigma T < t_2$. This yields the result $V[t_1+T + ( |\sfn| +1) \sigma T] \leq (1-\alpha^{(|\sfn|+1)T})^\sigma V[t_1 + T]$ when $t_2 < \infty$. 

If $t_2 = \infty$, then $\limfy{t} V[t] = \limfy{\sigma} V[t_1+T + ( |\sfn| +1) \sigma T] \leq (1-\alpha^{(|\sfn|+1)T})^\sigma V[t_1 + T]$. Note that $\alpha < 1$ implies $(1-\alpha^{(|\sfn|+1)T}) < 1$, and therefore the limit converges to zero. By \eqref{eq:minmaxv}, $\limfy{t} V[t] = 0$ implies $\limfy{t} \max_{i \in \sfn,\ l \in \Le^\N} |x_i[t] - x_l[t]| = 0$.
\end{proof}

\begin{remark}
Although the proof of Theorem \ref{thm:swmsr} follows a similar line of reasoning as the excellent results in \cite{saldana2017resilient}, Theorem \ref{thm:swmsr} contains two significant theoretical differences. First, Theorem \ref{thm:swmsr} considers the more general \emph{Byzantine} adversarial model \cite{leblanc2013resilient}, whereas the results in \cite{saldana2017resilient} consider only malicious adversaries.\footnote{In essence, malicious adversaries may update their state arbitrarily, but will send the same state information to all out-neighbors. Byzantine adversaries may update their state arbitrarily and send different information to different out-neighbors.} Second, Theorem \ref{thm:swmsr} considers consensus of the followers to a specific reference value propagated by the set of normally behaving leader agents which may lie outside the convex hull of initial agents' states. The analysis in \cite{saldana2017resilient} considers leaderless consensus to some unknown value in the convex hull of the initial normal agents' states.
\end{remark}

\section{Adversarial Implications}
\label{sec:adv}

We next discuss the adversarial implications of Theorem \ref{thm:swmsr}. In most \emph{leaderless} resilient consensus settings considered in prior work (e.g. \cite{saldana2017resilient}), the networks consist only of normally behaving agents seeking a common consensus value, and adversarial agents behaving arbitrarily. 
Often, these results guarantee resilient consensus if the adversary model is \emph{at most} $F$-local.
However, these results for leaderless resilient consensus raise the following critical question: \emph{What happens if the adversary model is NOT $F$-local?}
To the authors' best knowledge, little (if any) analysis has focused on the precise effects of the $F$-local assumption being violated in these scenarios. From a practical standpoint it is difficult to provide absolute guarantees that $\A$ will always be strictly $F$-local in any real-world application of resilient algorithms. It is therefore critical to understand the consequences which will occur if the $F$-local assumption does not hold.

Theorem \ref{thm:swmsr} can be used to show one possible catastrophic outcome if the $F$-local assumption is violated in a \emph{leaderless} network. More specifically, Theorem 1 can be used to demonstrate that for a leaderless network applying the SW-MSR algorithm, if there exists a colluding set of adversarial agents $\A$ and if the network is strongly $(T,t_0,2F+1)$-robust with respect to $\A$, then the adversarial agents can drive the states of \emph{all} normal agents \emph{to any arbitrary value}. This result is presented more precisely in the following corollary:
\begin{corollary}
 Let $\D[t] = (\V,\E[t])$ be a nonempty, nontrivial, simple digraph with $\Le = \emptyset$. Let $F \in \Z_+$, $t_0,t_1,t_2 \in \Z$ with $t_2 > t_1 \geq t_0 + T$.
Suppose that $\D[t]$ is strongly $(T,t_0,2F+1)$-robust w.r.t. a set of misbehaving agents $\A$ and all normally behaving agents $i \in \V \backslash \A$ apply the SW-MSR algorithm with parameter $F$. If all agents $j \in \A$ send a constant, common value $x_j^i[t]$ to all of their respective out-neighbors $i \in \V_j^{out}$ for all $t \in [t_1-T,t_2)$, and if $t_2 > t_1 + ( |\V \backslash \A| +1) \sigma T$ for some $\sigma \in \Z_+$, then
the error between the normally behaving agents' states and the adversaries' common state $x_j^i[t]$ is exponentially decreasing for $t \in [t_1-T,t_2)$. Furthermore, if $t_2 = \infty$ then $\limfy{t} \max_{i \in \V \backslash \A} |x_i[t] - x_j^i[t]| = 0$.
\end{corollary}

\begin{proof}
	The proof follows from Theorem \ref{thm:swmsr} by treating $\A$ as the set $\Le$, $\V \backslash \A$ as the set $\sfn$, and $x_j^i[t]$ as the signal $f_r[t]$. Note that by Definition \ref{def:misbehaving}, $x_j^i[t]$ need not be equal to any of the actual states $x_j[t]$ of $j \in \A$.
\end{proof}

In short, if the digraph $\D$ for a leaderless consensus network is strongly $(T,t_0,2F+1)$-robust w.r.t. the adversary set $\A$ and the adversaries collude to send a common constant to their out-neighbors on sufficiently long time intervals, the error between the normal agents and the adversarial signal will decrease exponentially. These conditions imply that the adversaries have the ability to drive the entire network to arbitrary state values.
When working with a given digraph $\D[t]$, this result demonstrates the need for awareness of the agent subsets $S$ such that $\D[t]$ is strongly $(T,t_0,2F+1)$-robust w.r.t. $S$. Adversaries seeking to obtain control of the network will succeed if such subsets are successfully compromised. 
However, determining methods to search for all such possible subsets $S$ is out of the scope of the current technical note. We leave exploration in this direction for future work.

\section{Simulations}
\label{sec:sim}

This paper presents a leader-follower framework which can tolerate up to $F$ arbitrarily misbehaving nodes. It can be applied to a wide range of problems where a network of agents need to be driven to a desired reference value by a set of leaders. Some examples  of such reference values include a reference altitude for unmanned aerial vehicles, a reference rendezvous time for multiple unmanned ground vehicles, and a reference radius for a circular patrolling path \cite{saldana2017resilient}, to name only a few.

The simulations consider agents in time-varying $k$-circulant digraphs. The Appendix contains the definition of $k$-circulant digraphs and details about the conditions under which $k$-circulant digraphs are strongly $r$-robust w.r.t. a subset. 
For each simulation the network topology switches between the three graphs depicted in Figure \ref{fig:SWMSRgraphs}. The union of the three graphs forms a $7$-circulant digraph.
The simulations consider the presence of malicious adversaries, which may send the \emph{same} misinformation to their respective out-neighbors \cite{leblanc2013resilient}.
In all simulations, agents have no knowledge as to whether their in-neighbors are normal, malicious, or behaving as leaders. In addition, $t_0 = 0$ and the agents' initial states are random values on the interval $[-25,25]$ for all agents in $(\V \backslash \Le)$.
	The results of the first simulation are shown in figure \ref{fig:one}.  In this simulation, the number of agents is 15, with $\Le = \{4,5,...,8\}$ (5 leaders). The time window is $T = 12$ steps, and the network switches graphs every 4 seconds ($\G_1,\G_2,\G_3,\G_1\ldots$), where the graphs are depicted in Figure \ref{fig:SWMSRgraphs}. By the results of the Appendix, the digraph is strongly $(12, 0, 5)$-robust w.r.t. $\Le$.
	For all normal follower agents, parameter $F = 2$. Two of the agents in the network behave maliciously. The function $f_r[t]$ is simply the constant $f_r[t] = 30$. 
	The error between the normally behaving agents' states (denoted by colored lines) and the normally behaving leaders' states (the solid black line) decreases exponentially in the presence of two misbehaving agents (the dotted red lines).
	The second simulation, depicted in Figure \ref{fig:two}, considers a scenario where $f_r[t]$ takes on different values over time. In this simulation, the network size is $30$ agents, with $\Le = \{1,2,\ldots,7\}$ (7 leaders). The time window for each agent is $T = 30$, and the network switches between graphs every 10 seconds ($\G_1,\G_2,\G_3,\G_1\ldots$). By the results of the Appendix, the digraph is strongly $(30, 0, 7)$-robust w.r.t. $\Le$. For all normal follower agents, parameter $F=3$. Three of the agents in the network behave maliciously. The error between the normally behaving agents and the normally behaving leaders decreases exponentially on the time intervals where $f_r[t]$ is constant as per the conditions of Theorem \ref{thm:swmsr}.

\begin{figure}
\centering
\includegraphics[width=.7\columnwidth]{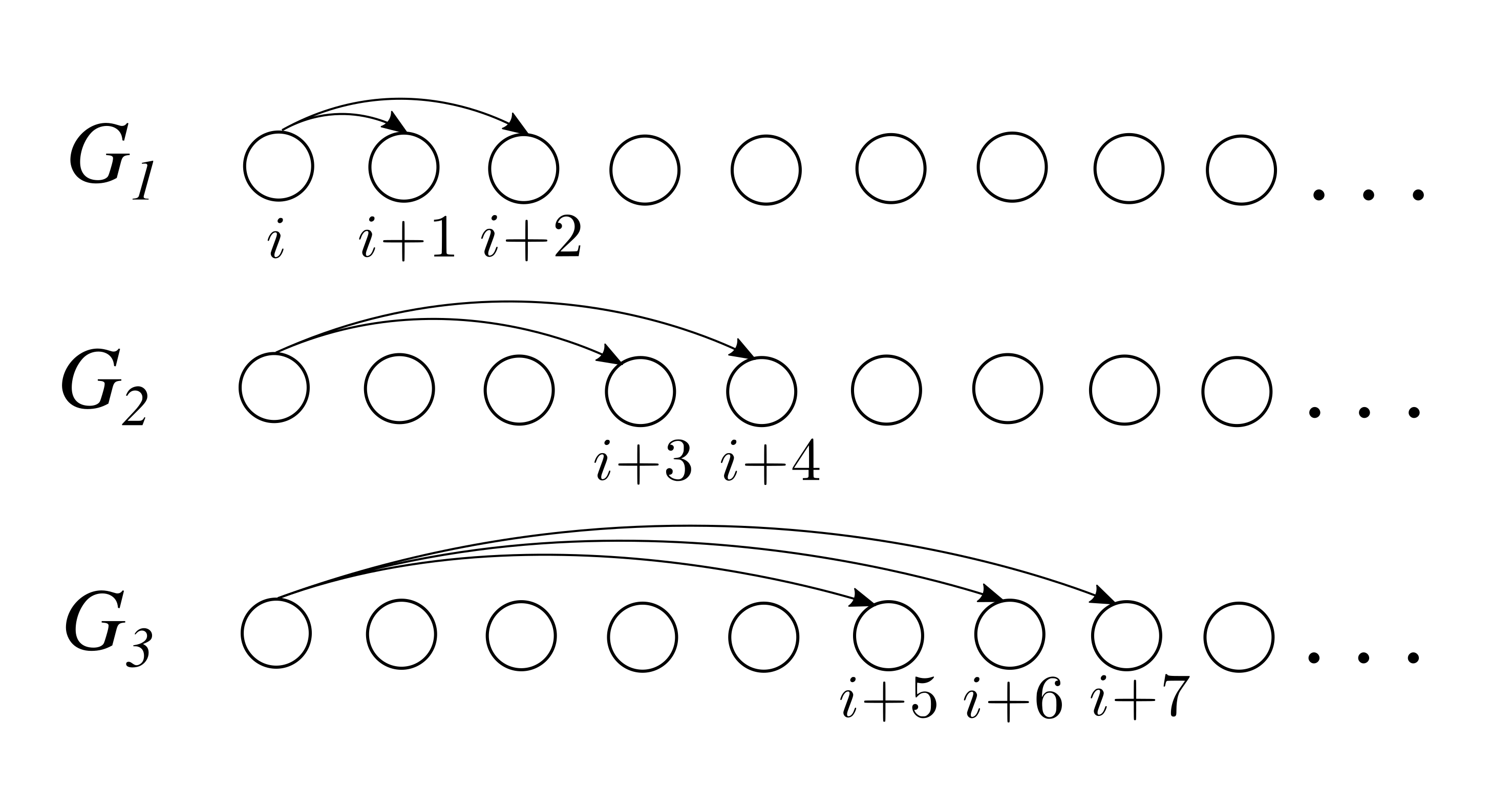}
\caption{Time-varying graphs used in the last two simulations. In each graph $\G_j$, $\forall i \in \V$ each agent $i$ sends its state information to the agents depicted.}
\label{fig:SWMSRgraphs}
\end{figure}

\begin{figure}
\centering
\includegraphics[width=.85\columnwidth]{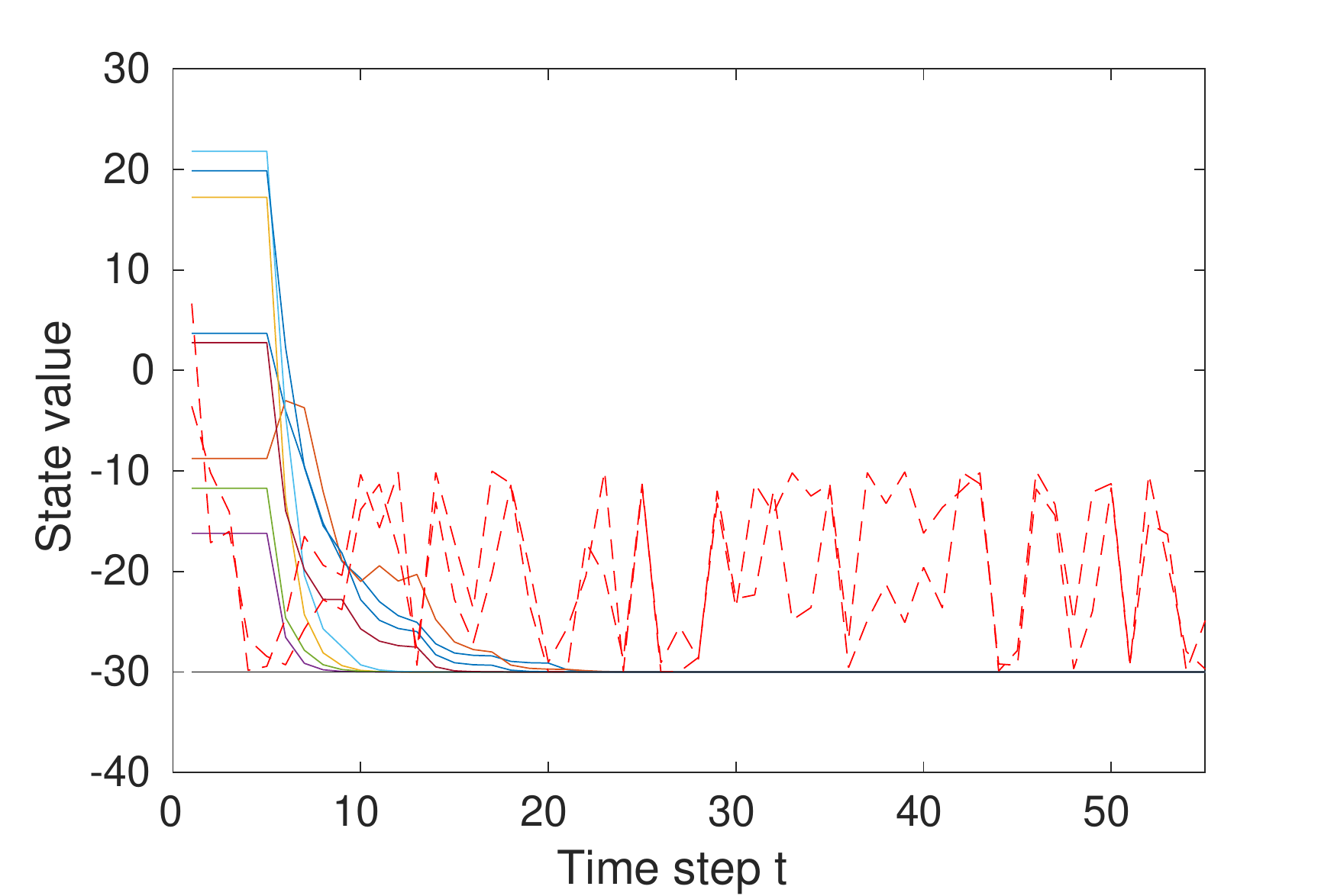}
\caption{Leader-follower simulation using the SW-MSR algorithm with a constant reference value in the presence of 2 malicious agents.}\label{fig:one}
\end{figure}

\begin{figure}
\centering
\includegraphics[width=.85\columnwidth]{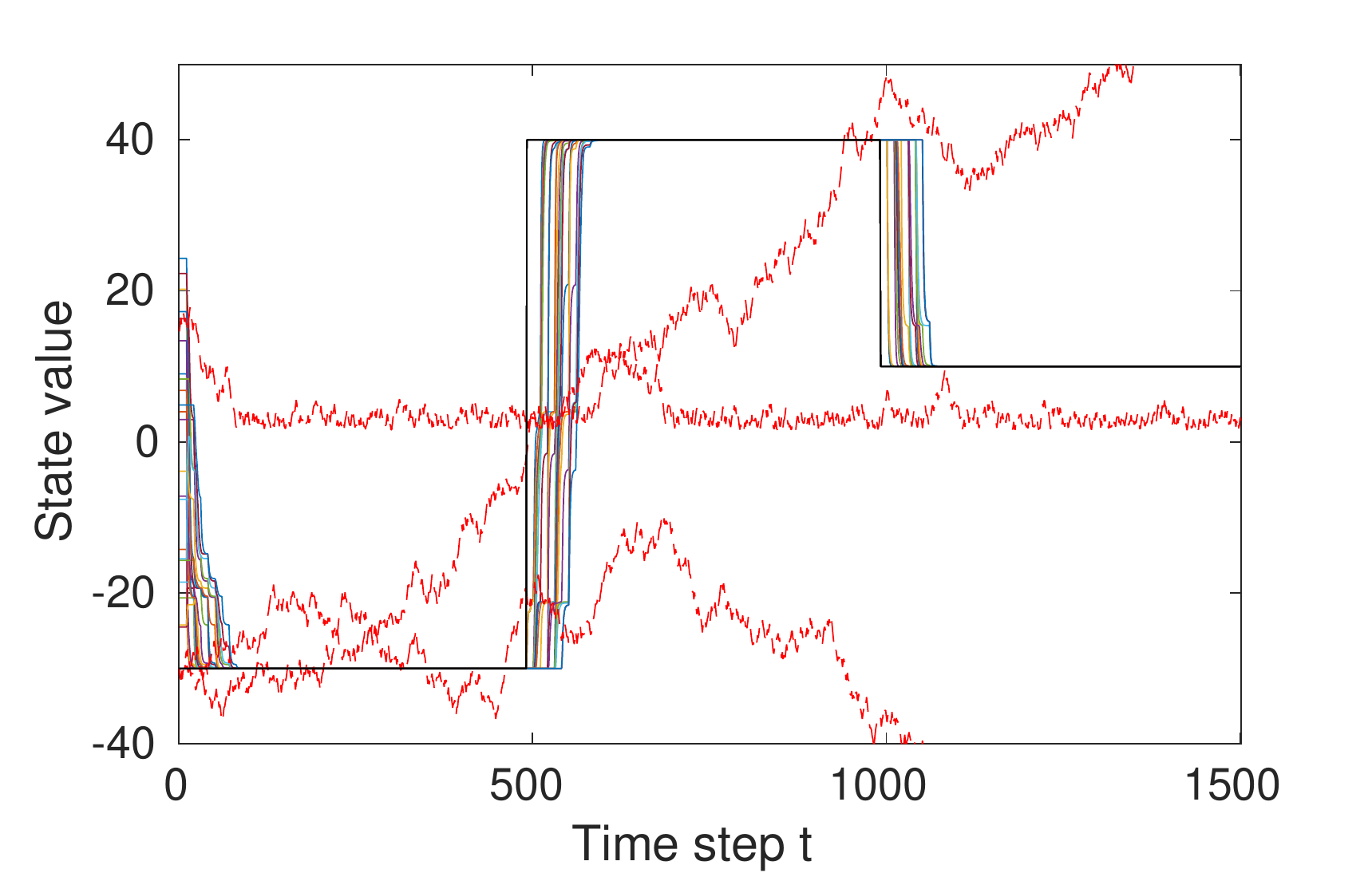}
\caption{Leader-follower simulation using the SW-MSR algorithm with a  time-varying reference value in the presence of 3 malicious agents. 
}\label{fig:two}
\end{figure}

\section{Conclusion}
\label{sec:concl}

This paper presented conditions for agents with discrete-time dynamics to resiliently track a reference signal propagated by a set of leader agents despite a bounded number of the leaders and followers behaving adversarially. Future work includes considering time-varying graphs, asynchronous communication, and also continuous-time systems.

\bibliographystyle{IEEEtran}

\bibliography{technote.bib}

\section{Appendix: $k$-Circulant Digraphs}
\label{appendix}

$k$-Circulant graphs \cite{usevitch2017r} are a particular class of graphs which, given a properly selected subset $S \subset \V$, are strongly $r$-robust w.r.t. $S$. To justify our choice of using $k$-circulant digraphs in our demonstrations, we provide formal conditions under which these graphs are strongly $r$-robust w.r.t. a given $S$. These results were first presented in our prior work \cite{usevitch2018resilient}.

First, we give the definitions of circulant graphs.
An undirected graph of $n$ nodes is called circulant if there exists a set $\{a_1, a_2, \ldots, a_m \in \mathbb{Z}_{\geq 0}: a_1 < a_2 < \ldots < a_m < n\}$ such that $\forall i \in \V$, $(i, \left[i \pm a_1 \right] \text{mod}\, n) \in \E, \ldots, (i,\left[i \pm a_m \right] \text{mod}\, n) \in \E$  \cite{boesch1984circulants}. We call such a graph an \emph{undirected circulant graph} and denote it as $C_n(\pm a_1, \pm a_2, \ldots, \pm a_m) = (\V, \E)$. The name \emph{circulant} derives from the adjacency matrix for such graphs being a circulant matrix \cite{boesch1984circulants,elspas1970graphs}.
These graphs are constructed over the additive group of integers modulo $n$ (the nodes $n+a$ and $a$ are congruent modulo $n$). Similarly, we call a digraph of $n$ nodes \emph{circulant} if there exists a set $\{a_1, a_2, \ldots, a_m: 0 < a_1 < a_2 < \ldots < a_m < n\},\ m \in \mathbb{Z}_{\geq 0}$ such that $\forall i \in \V$, $(i, \left[i+a_1 \right] \text{mod}\, n) \in \mathcal{E}, \ldots, (i,\left[i+a_m \right] \text{mod}\, n) \in \mathcal{E}$. We denote such a graph as $C_n(a_1, a_2, \ldots, a_m) = (\V, \E)$ and call it a \emph{directed circulant graph} or \emph{circulant digraph}.

$k$-Circulant digraphs and $k$-circulant undirected graphs are specific classes of circulant graphs defined as follows:
\begin{define}
Let $n \in \mathbb{Z},\ n \geq 2$ and let $k \in \mathbb{Z}: 1 \leq k < n-1$. A $k$-circulant digraph  is any circulant digraph of the form $C_n(1,2,3,\ldots,k) = (\V, \E)$.
\end{define}

\begin{define}
Let $n \in \mathbb{Z},\ n \geq 2$ and let $k \in \mathbb{Z}: 1 \leq k \leq \ceil{n/2}-1$. A $k$-circulant undirected graph  is any circulant graph of the form $C_n(\pm 1,\pm 2,\pm 3,\ldots,\pm k) = (\V, \E)$.
\end{define}

We show that these graphs can demonstrate strong $r$-robustness and TLF robustness with parameter $F$. As per the definition of circulant graphs, we assume all agents are indexed $1,\ldots,n$. In our next proof we refer to sets of consecutive agents by index. An example is $P_L = \{2,3,4,5,\ldots,9\}$ in a network of $n=15$ agents. Since the index numbers are defined on the set of integers modulo $n$, the set $P_L = \{14, 15,1,2\}$ would also be a set of consecutive agents in a network of $n=15$ agents.

\begin{theorem}
\label{thm:kcirc}
A $k$-circulant digraph $\D = C_n\{1,2,\ldots,k\}$ is strongly $r$-robust with respect to $\Le \subset \V$ if $\D$ contains a set of consecutive agents by index $P_L$ such that $|P_L| \leq k$ and $|P_L \cap \Le| \geq r$. 
\end{theorem}

\begin{proof}
Suppose $k \geq |P_L|$ and $|P_L \cap \Le| \geq r$. Without loss of generality, let the first agent in $P_L$ be labeled as agent $(n-|P_L|+1)$ and the last agent in $P_L$ as agent $n$. We must show that all nonempty $C \subseteq \V \backslash \Le$ are $r$-reachable. If agent $1 \in C$ then $C$ is $r$-reachable since $ \{(n-|P_L|+1),\ldots,n \} \subseteq \V_{1}$ which implies $|\V_{1} \cap (\V \backslash C)| \geq r$. Next, suppose that agent $1 \notin C$ and $2 \in C$. Since $\{(n-|P_L|+2),\ldots,1 \}\subseteq \V_{2}$, this implies that $|\V_{2} \cap (\V \backslash C)| \geq |\V_{1} \cap (\V \backslash C)| \geq r$ and therefore $C$ is $r$-reachable. This reasoning can be continued inductively by assuming $\{1,\ldots p-1\} \notin C$, $p \in C$, and observing that $|\V_{p} \cap (\V \backslash C)| \geq |\V_{p-1} \cap (\V \backslash C)|$. Analyzing the remaining subsets of this form in the graph yields the result. Note that if $p$ is ever the number of an agent in $\Le$, then we need not consider it ever being in $C$ and the analysis can be continued with the next agent not in $\Le$.
\end{proof}

Similar results also hold for undirected $k$-circulant graphs:

\begin{theorem}
An undirected circulant graph of the form $\mathcal{G} = C_n\{\pm 1,\pm 2,\ldots,\pm k\}$ is strongly $r$-robust with respect to $\Le \subset \V$ if $\mathcal{G}$ contains a set of consecutive agents $P_L$ such that $|P_L| \leq k$ and $|P_L \cap \Le| \geq r$. 
\end{theorem}

\begin{proof}
The same method as in Theorem \ref{thm:kcirc} can be applied to prove the result.
\end{proof}

\vspace{-3.3em}
\begin{biography}[{\includegraphics[trim={0 1in 0 0},width=.9in,height=1.25in,clip,keepaspectratio]{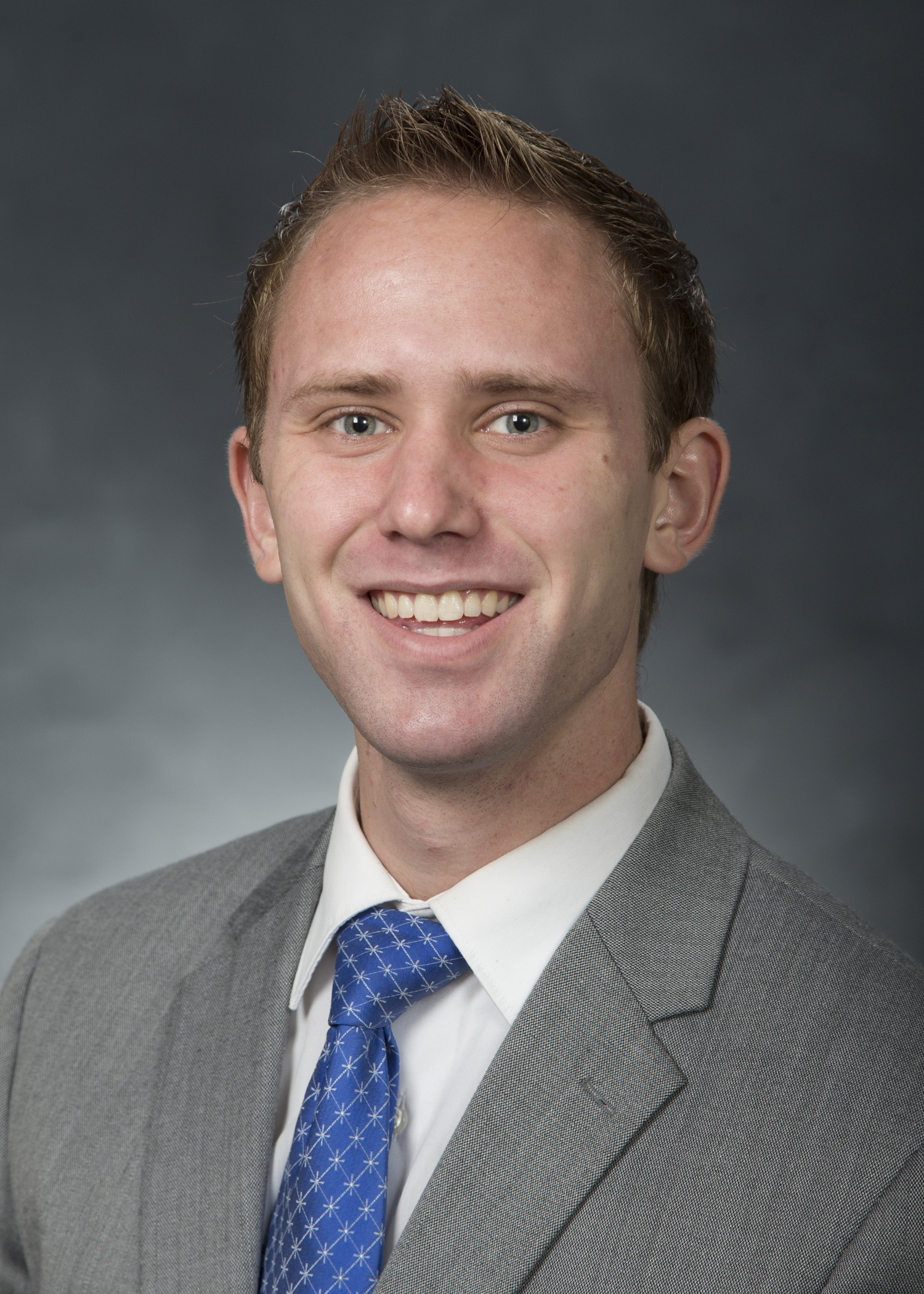}}]{James Usevitch} received his Bachelor's degree in Mechanical Engineering from Brigham Young University in 2016 and is currently pursuing a PhD in control theory in the Aerospace Engineering program at the University of Michigan. His research currently focuses on safety and security of multi-agent systems. He is a student member of the IEEE.
\end{biography}
\vspace{-4em}

\begin{biography}[{\includegraphics[width=1in,height=1.25in,clip,keepaspectratio]{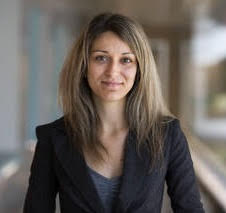}}]{Dimitra Panagou} received the Diploma and PhD degrees in Mechanical Engineering from the National Technical University of Athens, Greece, in 2006 and 2012, respectively. Since September 2014 she has been an Assistant Professor with the Department of Aerospace Engineering, University of Michigan. Prior to joining the University of Michigan, she was a postdoctoral research associate with the Coordinated Science Laboratory, University of Illinois, Urbana-Champaign (2012-2014), a visiting research scholar with the GRASP Lab, University of Pennsylvania (June 2013, fall 2010) and a visiting research scholar with the University of Delaware, Mechanical Engineering Department (spring 2009). Her research interests include the fields of multi-agent planning, coordination, control and estimation, with applications in safe and resilient unmanned aerial systems, robotic networks and autonomous multi-vehicle systems (ground, marine, aerial, space). She is a recipient of the NASA 2016 Early Career Faculty Award and of the Air Force Office of Science Research 2017 Young Investigator Award, and a member of the IEEE and the AIAA.
\end{biography}

\end{document}